\documentclass[preprint,authoryear,12pt]{elsarticle}

\usepackage{amssymb}
\usepackage{graphicx}
\usepackage{epsfig}
\usepackage{array}
\usepackage{mathrsfs}
\usepackage{amsfonts}
\usepackage{amsmath}
\usepackage{amssymb}
\usepackage{mathrsfs}
\usepackage{txfonts}
\usepackage{amsfonts}
\usepackage{amssymb}
\usepackage{amsmath}
\usepackage{mathrsfs}
\usepackage{amsmath,amssymb}
\usepackage{amsmath}
\input xy
\xyoption{all}

\newtheorem{theorem}{Theorem}
\newtheorem{lemma}[theorem]{Lemma}

\newdefinition{definition}{Definition}

\newtheorem{remark}{Remark}

\newproof{pf}{Proof}
\newproof{pot}{Proof of Theorem \ref{thm2}}

\journal{Automatica}

\begin{document}

\begin{frontmatter}


 \title{Structural Controllability of Switched Linear Systems\tnoteref{label1}}

 \author[Liu]{Xiaomeng Liu}
 \ead{liuxm1986@gmail.com}
 \author[Lin]{Hai Lin\corref{cor1}}
 \ead{hlin1@nd.edu}
  \author[Liu]{Ben M. Chen}
 \ead{bmchen@nus.edu.sg}
 \address[Liu]{Dept. of Electrical and Computer Engineering, National University of Singapore,
Singapore}
 \address[Lin]{Dept.
of Electrical Engineering, Univ of Notre Dame, Notre Dame, IN 46556, USA}
 \cortext[cor1]{Corresponding author. Hai Lin. tel 574-6313177 fax 574-6314393. Financial supports from NSF-CNS-1239222 and NSF-EECS-1253488 for this work are greatly acknowledged.
\thanks[{Manuscript titled ``Structural Controllability of Switched
Linear Systems'' is submitted to Automatica, 2013.}
 \cortext[cor2]{*}
 }

\begin{abstract}
In this paper, the structural controllability of switched linear
systems is investigated. The structural
controllability is a generalization of the traditional
controllability concept for dynamical systems, and purely based on
the graphic topologies among state and input vertices. First, two
kinds of graphic representations of switched linear systems are
proposed. Second, graph theory based necessary and sufficient
characterizations of the structural controllability for switched
linear systems are presented. Finally, the paper concludes with
illustrative examples and discussions on the results and future
work.
\end{abstract}


\end{frontmatter}

\section{Introduction}
%
%
%
%


As a special class of hybrid control systems, a switched linear
system consists of several linear subsystems and a rule that
orchestrates the switching among them. Switching between different
subsystems or different controllers can greatly enrich the control
strategies and may achieve better control performances than fixed
(non-switching) controllers (\cite{KJ,AWV}). Besides, switched
linear systems also have promising applications in control of
mechanical systems, aircrafts, satellites and swarming robots.
Driven by its importance in both theoretical research and practical
applications, switched linear system has attracted considerable
attention during the last decade, see e.g.,
\cite{linhai,sunzd,GL,D2}.

Much work has been done on the controllability of switched linear
systems. For example, the controllability and reachability for
low-order switched linear systems have been presented in
(\cite{KJI}). 
Complete geometric criteria for controllability and
reachability were established in
\cite{sunzd} and \cite{GL}. 

Up to now, all the previous work mentioned above has been based on
the traditional controllability concept of switched linear systems.
In this paper, we investigate the structural controllability of
a class of uncertain switched linear system, where the parameters of
subsystems' state matrices are either unknown or zero. This is a reasonable assumption as many system parameters are difficult to identify and only known to certain approximations. On the other
hand, we are usually pretty sure where zero elements are either by
coordination transformation or by the absence of physical connections among
components in the system. For example, in multi-agent systems, usually
only whether there is communication link between any two agents is
known, but the communication weights of linkages can not be
measured exactly.
Thus structural properties that are independent of a specific value
of unknown parameters, e.g., the structural controllability studied
here, are of particular interest.
A switched linear system is
said to be structurally controllable if one can find a set of values
for the unknown parameters such that the corresponding switched linear
system is controllable in the classical sense. For linear structured
systems, generic properties including structural controllability
have been studied extensively and it turns out that generic
properties including structural controllability are true for almost
all values of the parameters, see e.g., (\cite{lin,hm,RJ,KL,DCW,JW,KM,K,LD}).
This also holds true
for switched linear systems studied here and presents one of the
reasons why this kind of structural controllability is of interest. 


It turns out that the structural controllability of switched linear systems only depends on
graphic topologies among state and input vertices of individual subsystems and their union. The paper aims to characterize such a relationship, and its contribution is twofold. First, two kinds of graphic representations of switched linear systems are proposed. Second, graph theory based necessary and sufficient characterizations of the structural controllability for switched
linear systems are presented. Graphic conditions can help to understand how the graphic topologies
of dynamical systems influence the corresponding generic properties,
here especially for the structural controllability. This would be
helpful in many practical applications 
and motivates our pursuit on illuminating the structural
controllability of switched linear systems from a graph theoretical
point of view. Preliminary results of this paper appeared in
\cite{LLC2}.



The organization of this paper is as follows: In Section 2, we
introduce some basic preliminaries and the problem formulation,
followed by structural controllability study of switched linear
systems in Section 3, where several graphic necessary and
sufficient conditions for the structural controllability are given.
Illustrative examples together with discussions on a more general
case are also presented. Finally, some concluding remarks are drawn
in Section 4.


\section{Preliminaries and Problem Formulation}
\subsection{Graph Theory Preliminaries}


A matrix $P$ is said to be a structured matrix if its entries are either
fixed zeros or independent free parameters. $\tilde{P}$ is called
admissible (with respect to $P$) if it can be obtained by fixing the
free parameters of $P$ at some particular values. In addition
$P_{ij}$ is adopted to represent the element of $P$ from row $i$ and
column $j$.


Consider a linear control system:
\begin{equation}\label{l1}
\dot{x}=Ax(t)+Bu(t),
\end{equation}
where $x(t)\in \mathbb{R}^n$ and $u(t)\in \mathbb{R}^r$. The
matrices $A$ and $B$ are assumed to be structured matrices, which means that their
elements are either fixed zeros or free parameters. This structured
system given by matrix pair $(A, B)$ can be described by a directed
graph (\cite{lin}).

The representation graph of structured system $(A, B)$ is a directed graph
$\mathcal{G}$, with vertex set $\mathcal{V}=\mathcal{X}\cup
\mathcal{U}$, where $\mathcal{X}=\{x_1,x_2,\ldots,x_n\}$, which is
called $state~vertex~ set$ and $\mathcal{U}=\{u_1,u_2,\dots,u_r\}$,
which is called $input~vertex~ set$, and edge set
$\mathcal{I}=\mathcal{I}_{UX}\cup\mathcal{I}_{XX}$, where
$\mathcal{I}_{UX}=\{(u_i,x_j)|B_{ji}\neq 0, 1\leq i \leq r, 1\leq
j\leq n\}$ and $\mathcal{I}_{XX}=\{(x_i,x_j)|A_{ji}\neq 0, 1\leq i
\leq n, 1\leq j\leq n\}$ are the oriented edges between inputs and
states and between states defined by the interconnection matrices
$A$ and $B$ above. This directed graph (for notational simplicity,
we will use digraph to refer to directed graph) $\mathcal{G}$ is
also called the graph of matrix pair $(A,B)$ and denoted by
$\mathcal{G}(A,B)$. The following notations from \cite{lin} are recalled.



\begin{definition}\label{def14}(\textit{Stem})
An alternating sequence of distinct vertices and oriented edges is
called a directed path, in which the terminal node of any edge never
coincide to its initial node or the initial or the terminal nodes of
the former edges. A stem is a directed path in the state vertex set
$\mathcal{X}$, that begins in the input vertex set $\mathcal{U}$.
\end{definition}

\begin{definition}(\textit{Accessibility}) A vertex (other than the input vertices) is called
{\em nonaccessible} if and only if there is no possibility of reaching
this vertex through any stem of the graph $\mathcal{G}$.
\end{definition}
\begin{definition}\label{def17}(\textit{Dilation})
Consider one vertex set $S$ formed by the vertices from the state
vertices set $\mathcal{X}$ and determine another vertex set $T(S)$,
which contains all the vertices $v$ with the property that there
exists an oriented edge from $v$ to one vertex in $S$. Then the
graph $\mathcal{G}$ contains a `{\em dilation}' if and only if there
exist at least a set $S$ of $k$ vertices in the vertex set of the
graph such that there are no more than $k-1$ vertices in $T(S)$.
\end{definition}

\subsection{Switched Linear System, Controllability and Structural Controllability}

In general, a switched linear system is composed of a family of
subsystems and a rule that governs the switching among them, and is
mathematically described by
\begin{eqnarray}\label{eq2}
\dot x(t)=&A_{\sigma(t)} x(t)+B_{\sigma(t)} u(t) ,
\end{eqnarray}
where $x(t)\in \mathbb{R}^n$ are the states, $u(t)\in
\mathbb{R}^{r}$ are piecewise continuous input,
$\sigma:[0,\infty)\rightarrow M \triangleq \{1,\ldots,m\}$ is the switching signal. System
(\ref{eq2}) contains $m$ subsystems $(A_i,B_i),$ $i\in
\{1,\ldots,m\}$ and $\sigma(t)$= $i$ implies that the
$i$th subsystem $(A_i,B_i)$ is active at time instance $t$.

In the sequel, the following definition of controllability of system
(\ref{eq2}) will be adopted (\cite{sunzd}):

\begin{definition}\label{def1} Switched linear system (\ref{eq2}) is said to be (completely) controllable
if for any initial state $x_0$ and final state $x_f$, there exist a
time instance $t_f>0$, a switching signal $\sigma:[0,t_f)\rightarrow
M$ and an input $u:[0,t_f)\rightarrow \mathbb{R}^{r}$ such that
$x(0)=x_0$ and $x(t_f)=x_f$.
\end{definition}


For the controllability of switched linear systems, a matrix rank condition was given in \cite{sunzd}.

\begin{lemma}\label{lem1}\rm If the matrix:\begin{equation}\begin{split}\label{eq3} &[B_1,B_2,\ldots,B_m,
A_1B_1,A_2B_1,\ldots,A_mB_1,A_1B_2,A_2B_2,\ldots,A_mB_2,\ldots,
A_1B_m,A_2B_m,\\&\ldots,A_mB_m,A_1^2B_1,A_2A_1B_1,\ldots,A_mA_1B_1,A_1A_2B_1,A_2^2B_1,\ldots,A_mA_2B_1,\ldots,A_1A_mB_m\\&,A_2A_mB_m,\ldots,A_m^2B_m,
\\&,\ldots,\\&A_1^{n-1}B_1,A_2A_1^{n-2}B_1,\ldots,A_mA_1^{n-2}B_1,A_1A_2A_1^{n-3}B_1,A_2^2A_1^{n-3}B_1,\ldots,A_mA_2A_1^{n-3}B_1\ldots,\\&A_1A_m^{n-2}B_m,
A_2A_m^{n-2}B_m\ldots, A_m^{n-1}B_m] \end{split}\end{equation}has
full row rank $n$, then switched linear system (\ref{eq2}) is
controllable, and vice versa.
\end{lemma}

\begin{remark}\rm This matrix is called controllability matrix of switched linear system (\ref{eq2}) and denoted as $\mathcal{C}$$(A_1,\ldots,A_m,B_1,\ldots,B_m)$. If we use $Im$ $P$ to represent the range space of
arbitrary matrix $P$, then $Im$ $\mathcal{C}$$(A_1,\ldots,A_m,B_1,\ldots,B_m)$ is the
controllable subspace of switched linear system (\ref{eq2})(\cite{sunzd}). The above lemma implies that system
(\ref{eq2}) is controllable if and only if
$Im$ $\mathcal{C}$$(A_1,\ldots,A_m,B_1,\ldots,B_m)=\mathcal{R}^n$.
Besides, controllable subspace can be expressed as $\langle
A_1,\ldots,A_m$ $|B_1,\ldots,B_m \rangle$, which is the smallest
subspace containing $Im$$B_i$, $i=1,\ldots,m$ and invariant under
the transformations $A_1,\ldots,A_m$ (\cite{D2}).
\end{remark}

In view of structural controllability, system (\ref{eq2}) will be
treated as structured switched linear system defined as:

\begin{definition}For structured system (\ref{eq2}), elements of all the
matrices $(A_1, B_1,$ $\ldots,$ $ A_m, B_m)$ are either fixed zero or
free parameters and free parameters in different subsystems
$(A_i,B_i), i\in M$ are independent. A numerically given matrices
set $(\tilde{A}_1, \tilde{B}_1,\ldots, \tilde{A}_m, \tilde{B}_m)$ is
called an admissible numerical realization (with respect to $(A_1,
B_1,\ldots, A_m, B_m)$) if it can be obtained by fixing all free
parameter entries of $(A_1, B_1,\ldots, A_m, B_m)$ at some
particular values.
\end{definition}

Similar with the definition of structural controllability of linear
system in \cite{K}, we have the following definition for structural
controllability of switched linear system (\ref{eq2}):

\begin{definition}\label{def2} Switched linear system (\ref{eq2}) given by its structured matrices $(A_1, B_1,\ldots,$ $A_m, B_m)$
is said to be structurally controllable if and only if there exists
at least one admissible realization $(\tilde{A}_1,
\tilde{B}_1,\ldots, \tilde{A}_m, \tilde{B}_m)$ such that the
corresponding switched linear system is controllable in the usual
numerical sense.
\end{definition}
\begin{remark}\rm
It turns out that once a structured system is controllable for one
choice of system parameters, it is controllable for almost all
system parameters, in which case the structured system then will be
said to be structurally controllable (\cite{lin}, \cite{DCW}).
\end{remark}
 Before proceeding further, we need to introduce the
definition of $g$-rank:
\begin{definition} The generic rank ($g$-rank) of a
structured matrix $P$ is defined to be the maximal rank that $P$
achieves as a function of its free parameters.
\end{definition}

Then, we have the following algebraic condition for structural
controllability:

\begin{lemma}\label{lem2} \rm Switched linear system
(\ref{eq2}) is structurally controllable if and only if $g$-rank
$\mathcal{C}$$(A_1,\ldots,A_m,B_1,\ldots,B_m)$ = $n$.
\end{lemma}
\section{Structural Controllability of Switched Linear Systems}

\subsection{Criteria Based on Union Graph}

For switched linear system (\ref{eq2}), digraph $\mathcal
{G}_i(A_i,B_i)$ with vertex set $\mathcal{V}_i$ and edge set
$\mathcal{I}_i$ can be adopted as the representation graph of its
subsystems $(A_i,B_i)$, $i \in \{1,\ldots,m\}.$ Switched linear system (\ref{eq2}) can be represented by a union graph $\mathcal{G}$ (actually a digraph) of these digraphes $\mathcal
{G}_i(A_i,B_i)$.


\begin{definition}\label{def3}  Given a collection of digraphes $\mathcal{G}_i = \{\mathcal{V}_i, \mathcal{I}_i\}$, their union graph is
\begin{eqnarray}
\mathcal{G}_1\cup \mathcal{G}_2\cup\ldots\cup
\mathcal{G}_m=\{\mathcal{V}_1\cup \mathcal{V}_2\cup \ldots \cup
\mathcal{V}_m;\mathcal{I}_1\cup \mathcal{I}_2\cup \ldots\cup
\mathcal{I}_m\}.
\end{eqnarray}
\end{definition}

\begin{remark}\label{rem3} \rm It turns out that union
graph $\mathcal{G}$ is the representation graph of linear structured
system: $(A_1+A_2+\ldots+A_m,B_1+B_2+\ldots+B_m)$. The reason is as
follows: If the element at position $a_{ji}(b_{ji})$ in matrix
$[A_1+A_2+\ldots+A_m, B_1+B_2+\ldots+B_m]$ is a free parameter, this
implies that there exist some matrices $[A_p, B_p]$, $p=1,\ldots,m$
such that the element at position $a_{ji}(b_{ji})$ is also a free
parameter and in the corresponding subgraph $\mathcal{G}_p$, there
is an edge from vertex $i$ to vertex $j$. According to the
definition of union graph, it follows that there is also an edge
from vertex $i$ to vertex $j$ in union graph $\mathcal{G}$. If the
element at position $a_{ji}(b_{ji})$ in $[A_1+A_2+\ldots+A_m,
B_1+B_2+\ldots+A_m]$ is zero, this implies that for every matrices
$[A_p, B_p]$, $p=1,\ldots,m$, the element at position
$a_{ji}(b_{ji})$ is zero and in the corresponding subgraph
$\mathcal{G}_p$, there is no edge from vertex $i$ to vertex $j$. It
follows that there is also no edge in union graph $\mathcal{G}$ from
vertex $i$ to vertex $j$.
\end{remark}


\begin{definition}\label{def5}(\cite{lin}) The matrix pair $(A,B)$ is said to be
reducible or of form I if there exists a permutation matrix $P$ such
that they can be written in the following form:
$
PAP^{-1}=\left[
\begin{array}{ccc}
A_{11}&0\\
A_{21}&A_{22}\\
\end{array}\right],PB=\left[
\begin{array}{cc}
0\\
B_{22}\\
\end{array}\right],
$where $ A_{11}\in \mathbb{R}^{p \times p}$ , $A_{21 } \in
\mathbb{R}^{(n - p) \times p}$,$ $ $  A_{22}  \in \mathbb{R}^{(n -
p) \times (n- p)} $ and $ B_{22}  \in \mathbb{R}^{(n- p) \times r}$.
\end{definition}
\begin{remark}\label{rem4} \rm Whenever the matrix pair $(A,B)$ is of form I, the system is
structurally uncontrollable (\cite{lin}) and meanwhile, the
controllability matrix
$C\triangleq\left[B,AB,\ldots,A^{n-1}B\right]$ will have at least
one row which is identically zero for all parameter values
(\cite{KL}). If there is no such permutation matrix $P$, we say that
the matrix pair $(A,B)$ is irreducible.
\end{remark}

\begin{definition}\label{def6} (\cite{lin}) The matrix pair $(A,B)$ is said to be of
form II if there exists a permutation matrix $P$ such that they can
be written in the following form:
$
\left[PAP^{-1},PB\right]=\left[
\begin{array}{ccc}
P_1\\
P_2
\end{array}\right],
$ where $P_2\in \mathbb{R}^{(n-k)\times (n+r)}$ , $P_1 \in
\mathbb{R}^{k \times (n+r)}$ with no more than $k-1$ nonzero columns
(all the other columns of $P_1$ have only fixed zero entries).
\end{definition}

The following lemma characterizes the structural controllability for linear system $(A,B)$ (\cite{lin,K}):
\begin{lemma}\label{lem3}\rm (\cite{lin,K}) For linear structured system (\ref{l1}), the
following statements are equivalent:
\begin{enumerate}
\item[a)] the pair $(A,B)$ is structurally controllable;
\item[b)]i) $[A,B]$ is irreducible or not of form I,\\
         ii) $[A,B]$ has $g$-rank$[A,B]=n$ or is not of form II;
\item[c)]i) there is no nonaccessible vertex in $\mathcal{G}(A,B)$,\\
         ii) there is no `dilation' in $\mathcal{G}(A,B)$.
\end{enumerate}
\end{lemma}

This lemma proposed interesting graphic conditions for structural
controllability of linear systems and revealed that the structural
controllability is totally determined by the underlying graph
topology. Next, we turn to the switched linear system (\ref{eq2})  and prove a graphic sufficient condition for its structural controllability.

%

\begin{theorem}\label{the1}\rm Switched linear system (\ref{eq2}) with graphic topologies
$\mathcal{G}_i$, $i\in \{1,\ldots,m\}$, is structurally controllable
if its union graph $\mathcal{G}$ satisfies:
\begin{enumerate}
\item[i)]there is no nonaccessible vertex in $\mathcal{G}$,
\item[ii)]there is no `dilation' in $\mathcal{G}$.
\end{enumerate}
\end{theorem}
\begin{pf}
Assume the two conditions in this theorem are satisfied. According
to Remark \ref{rem3} and Lemma \ref{lem3}, the corresponding linear
system $(A_1+A_2+\ldots +A_m, B_1+B_2+\ldots +B_m)$ is structurally
controllable. It follows that there exist some scalars for the free
parameters in matrices $(A_i, B_i),i=1,2,\ldots,m$ such that
controllability matrix
\begin{equation}\begin{split}
&[B_1+B_2+\ldots+B_m,(A_1+A_2+\ldots+A_m)( B_1+B_2+\ldots+B_m),\\&
(A_1+A_2+\ldots+A_m)^2(B_1+B_2+\ldots+B_m),~\ldots,\\&(A_1+A_2+\ldots+A_m)^{n-1}(
B_1+B_2+\ldots+B_m)]\nonumber\end{split}\end{equation} has full row
rank $n$. Expanding the matrix, it follows that matrix
\begin{equation}\begin{split}[&B_1+B_2+\ldots+B_m, A_1B_1+A_2B_1+\ldots+A_mB_1+A_1B_2+A_2B_2\\&
+\ldots+A_mB_2+\ldots
+A_1B_m+A_2B_m\ldots+A_mB_m,~\ldots,~\\&A_1^{n-1}B_1
+A_2A_1^{n-2}B_1+\ldots+A_m^{n-1}B_m]\nonumber\end{split}\end{equation}
has full rank $n$. \\The following matrix can be got after adding
some column vectors to the above matrix:
\begin{equation}\begin{split}
&[B_1+B_2+\ldots+B_m,
B_2,\ldots,B_m,A_1B_1+A_2B_1+\ldots+A_mB_1+A_1B_2+A_2B_2\\&+\ldots+A_mB_2+\ldots+
A_1B_m+A_2B_m+\ldots+A_mB_m,A_2B_1,\ldots,A_mB_m,~\ldots ,\\&
A_1^{n-1}B_1
+A_2A_1^{n-2}B_1+\ldots+A_1A_m^{n-2}B_1+\ldots+A_m^{n-1}B_m,
A_2A_1^{n-2}B_1,\ldots, \\&A_1A_m^{n-2}B_1,\ldots,
A_m^{n-1}B_m].\nonumber\end{split}\end{equation} Since this matrix
still has $n$ linear independent column vectors, it follows that it
has full row rank $n$. Next, subtracting $B_2,\ldots,B_m$ from $
B_1+B_2+\ldots+B_m$; subtracting $A_2B_1,\ldots,A_mB_m$ from $
A_1B_1+A_2B_1+\ldots+A_mB_1+\ldots+A_1B_m+\ldots+A_mB_m $ and
subtracting $A_2A_1^{n-2}B_1,\ldots,A_1A_m^{n-2}B_1,\ldots,
A_m^{n-1}B_m$ from $A_1^{n-1}B_1
+A_2A_1^{n-2}B_1+\ldots+A_1A_m^{n-2}B_1+\ldots+ A_m^{n-1}B_m$, we
can get the following matrix:\begin{equation}\begin{split}
&[B_1,B_2,\ldots,B_m,A_1B_1,A_2B_1,\ldots,A_mB_m,~\ldots,
\\&A_1^{n-1}B_1,A_2A_1^{n-2}B_1,\ldots,A_1A_m^{n-2}B_1,\ldots,
A_m^{n-1}B_m],\nonumber \end{split}\end{equation} which is the
controllability matrix for switched linear systems (\ref{eq2}).
Since column fundamental transformation does not change the matrix
rank, this matrix still has full row rank $n$. Hence, the switched
linear system (\ref{eq2}) is structurally controllable. 
\end{pf}

Actually, from the proof, we can see that full rank of
controllability matrix of linear system
$(A_1+A_2+\ldots+A_m,B_1+B_2+\ldots+B_m)$ in Remark \ref{rem3}
implies the full rank of controllability matrix of system
(\ref{eq2}), which means that the structural controllability of this
linear system implies structural controllability of system
(\ref{eq2}). It turns out that this criterion is not necessary for
system (\ref{eq2}) to be structurally controllable (see the example
in subsection 3.4). This implies that the union graph does not
contain enough information for determining structural
controllability. This is because edges from different subsystems are
not differentiated in union graph. In the following subsection,
another graphic representation of switched linear systems is
proposed, from which necessary and sufficient conditions for
structural controllability arise.

\subsection{Criteria Based on Colored Union Graph}

In the union graph, there is no distinction made between the edges from different subsystems. To  solve this issue, we introduce the following \textit{`colored union graph'} as another graphic representation of switched systems.
\begin{definition}\label{def4}
 Given a collection of digraphes $\mathcal{G}_i = \{\mathcal{V}_i, \mathcal{I}_i\}$, their colored union graph is $\mathcal{\tilde{G}}(\mathcal{\tilde{V}},\mathcal{\tilde{I}})$, where its vertex set $\mathcal{\tilde{V}}=\{\mathcal{V}_1\cup
\mathcal{V}_2\cup \ldots \cup \mathcal{V}_m$\} and edge set
$\mathcal{\tilde{I}}=\{e|e\in \mathcal{I}_i, i=1,2,\ldots,m\}$,
i.e., for $i\in \{1,\ldots, m\}$.

\end{definition}

Intuitively, each edge $e$ in the colored union graph $\mathcal{\tilde{G}}$ is associated
an index $i$ (color) to indicate that $e$ comes from the $i$th subsystem (subgraph $\mathcal{G}_i$). With this colored union graph, several graphic properties are
introduced in the following lemmas.
\begin{lemma}\label{lem4} \rm There is no nonaccessible vertex in the colored union graph
$\mathcal{\tilde{G}}$ of switched linear system (\ref{eq2}) if and
only if the matrix $ [A_1+A_2+\cdots+A_m, B_1+B_2+\cdots+B_m] $ is
irreducible or not of form I.
\end{lemma}

\begin{pf} One vertex is accessible if and only if it can be reached by a stem. From Definitions \ref{def3} and \ref{def4}, it follows that there is no
nonaccessible vertex in the colored union graph if and only if there
is no nonaccessible vertex in the union graph. Besides, from Remark
\ref{rem3}, it is clear that the matrix representation of the union
graph is $[A_1+A_2+\cdots+A_m, B_1+B_2+\cdots+B_m]. $ According to
Lemma \ref{lem3}, there is no nonaccessible vertex in the union
graph if and only if matrix is irreducible or not of form I.
Consequently the equivalence between accessibility of colored union
graph and irreducibility of this matrix gets proved.
\end{pf}

A new graphic property `$S$-$dilation$' in colored union graph is
introduced here:

\begin{definition}\label{def7} In the colored union graph
$\mathcal{\tilde{G}}$, which is composed of subgraphs
$\mathcal{G}_i$, $i=1,2,\ldots,m$, consider one vertex set $S$
formed by the vertices from the state vertex set $\mathcal{X}$ and
determine another vertex set $T(S)=\{v|v\in T_i(S),
i=1,2,\ldots,m\}$, where $T_i(S)$ is a vertex set in $\mathcal{G}_i$
which contains all the vertices $w$ with the property that there
exists an oriented edge from $w$ to one vertex in $S$. Then
$|T(S)|=\sum_{i=1}^m |T_i(S)|$. If $|T(S)|<|S|$, we say that there
is a $S$-$dilation$ in the colored union graph
$\mathcal{\tilde{G}}$.
\end{definition}

Based on this new graphic property, the following lemma can be
introduced:

\begin{lemma}\label{lem5} \rm There is $S$-$dilation$ in the colored union graph
$\mathcal{\tilde{G}}$ of switched linear system (\ref{eq2}) if and
only if matrix $[A_1,A_2,\ldots,A_m, B_1,B_2,\ldots,B_m]$ is of form
II. It means that this matrix can be written into:
$[A_1,A_2,\ldots,A_m, B_1,B_2,\ldots,B_m]$=$\left[
\begin{array}{ccc}
P_1\\
P_2
\end{array}\right]$,
 where $P_1 \in \mathbb{R}^{p \times k}$ with no more than $p-1$
nonzero columns (all the other columns of $P_1$ have only fixed zero
entries).
\end{lemma}

\begin{pf}From \cite{lin} and \cite{hm} or Lemma \ref{lem3}, it is
known that in linear systems, there is no `dilation' in the
corresponding graph if and only if the matrix pair $[A,B]$ can not
be of form II or have $g$-rank = $n$. From the explanation of this
result in \cite{lin} and Definition \ref{def6}, $P_1$ in $[A,B]$ has
$p$ rows, which actually represents the $p$ vertices of vertex set
$S$ (defined for dilation), and each nonzero element of each row of
$P_1$ represents that there is one vertex pointing to the vertex
presented by this row. Therefore, the number of nonzero columns in
$P_1$ is the number of vertices pointing to some vertex in $S$, and
actually equals to $|T(S)|$. Furthermore, by the definition of
$S$-$dilation$, $|T(S)|$ is now the summation of $|T_i(S)|$, $i\in
\{1,\ldots,m\}$, in every subgraph. It follows that there is
$S$-$dilation$ in $\mathcal{\tilde{G}}$ if and only if matrix
$[A_1,A_2,\ldots,A_m, B_1,B_2,\ldots,B_m]$ is of form II.
\end{pf}

Before going further to give another algebraic explanation of
$S$-$dilation$, one definition and lemma proposed in \cite{RJ} must
be introduced first:

\begin{definition}\label{def8}(\cite{RJ}) A structured
$n\times m^\prime$ $(n\leq m^\prime)$ matrix $A$ is of form $(t)$
for some $t$, $1\leq t\leq n$, if for some $k$ in the range
$m^\prime-t<k\leq m^\prime$, $A$ contains a zero submatrix of order
$(n+m^\prime-t-k+1)\times k$.
\end{definition}

\begin{lemma}\label{lem6}\rm (\cite{RJ})~~ $g$-rank of $A=t$
\begin{enumerate}
\item[i)] for $t=n$ if and only if $A$ is not of form $(n)$;
\item[ii)]for $1\leq t< n$ if and only if $A$ is of form $(t+1)$ but
not of form $(t)$.
\end{enumerate}
\end{lemma}

From the above definition and lemma, another lemma is proposed here:

\begin{lemma}\label{lem7}\rm There is no $S$-$dilation$ in the colored union graph
$\mathcal{\tilde{G}}$ of switched linear system (\ref{eq2}) if and
only if the following matrix $[A_1,A_2,\ldots,A_m,
B_1,B_2,\ldots,B_m]$ has g-rank $n$.
\end{lemma}

\begin{pf}
\textit{~~Necessity:} If this matrix has $g$-rank $< n$, from Lemma
\ref{lem6}, it follows that this matrix is of form $(n)$. Then
referring to Definition \ref{def8}, the matrix must have a zero
submatrix of order $(n+m^\prime-t-k+1)\times k$. Here, $t$ can be
chosen as $n$, then matrix has a zero submatrix of order
$(m^\prime-k+1)\times k$. For this $(m^\prime-k+1)$ rows, there are
only $(m^\prime-k)$ nonzero columns. Consequently, the matrix is of
form II and by Lemma \ref{lem5}, there is $S$-$dilation$ in the
colored union graph $\mathcal{\tilde{G}}$ of switched linear system
(\ref{eq2}).

\textit{Sufficiency}: If there is $S$-$dilation$ in the colored
union graph $\mathcal{\tilde{G}}$, by Lemma \ref{lem5}, the matrix
is of form II, then obviously $P_1$ in this matrix can not have row
rank equal to $k$ and furthermore, this matrix can not have $g$-rank
= $n$.
\end{pf}

With the above definitions and lemmas, a graphic necessary and
sufficient condition for switched linear system to be structurally
controllable can be proposed:

 \begin{theorem}\label{the2}\rm
Switched linear system (\ref{eq2}) with graphic representations
$\mathcal{G}_i$, $i\in \{1,\ldots,m\}$, is structurally
controllable if and only if its colored union graph
$\mathcal{\tilde{G}}$ satisfies the following two conditions:
\begin{enumerate}
\item[i)]there is no nonaccessible vertex in the colored union graph
$\mathcal{\tilde{G}}$,
\item[ii)]there is no $S$-$dilation$ in the colored union graph $\mathcal{\tilde{G}}$.
\end{enumerate}
\end{theorem}

\begin{pf} \textit{~~Necessity:} \textit{(i)} If there exist nonaccessible vertices in
$\mathcal{\tilde{G}}$, by Lemma \ref{lem4}, the matrix
$[A_1+A_2+\cdots+A_m, B_1+B_2+\cdots+B_m]$ is reducible or of form
I. It follows that the controllability matrix
\begin{equation}\begin{split}[
&B_1+B_2+\ldots+B_m,(A_1+A_2+\ldots +A_m)(
B_1+B_2+\ldots+B_m),\\&(A_1+A_2+\ldots +A_m)^2(
B_1+B_2+\ldots+B_m),~\ldots,\\&(A_1+A_2+\ldots+A_m)^{n-1}(
B_1+B_2+\ldots+B_m)]\nonumber
\end{split}\end{equation}
always has at least one row that is identically zero (Remark
\ref{rem4}). It is clear that every component of the matrix, such as
$B_i,A_iB_j~and ~A_i^pA_j^qB_r$ has the same row always to be zero.
As a result, the controllability matrix
\begin{equation}\begin{split} &[B_1,\ldots,B_m,
A_1B_1,\ldots,A_mB_1,\ldots,A_mB_m,A_1^2B_1,\ldots,A_mA_1B_1,\ldots,A_1^2B_m,\ldots,
\\&A_mA_1B_m,\ldots,A_1^{n-1}B_1,\ldots,A_mA_1^{n-2}B_1,\ldots,A_1A_m^{n-2}B_m,\ldots,
A_m^{n-1}B_m] \end{split}\nonumber\end{equation} always has one zero
row and can not be of full rank $n$. Therefore, switched linear
system (\ref{eq2}) is not structurally controllable.

\textit{(ii)} Suppose that switched linear system (\ref{eq2}) is
structurally controllable, i.e., the controllability matrix
satisfies $g$-rank $\mathcal{C}$$(A_1,\ldots,A_m,B_1,\ldots,B_m)=n$.
Specifically, $\textit{Im}[B_1,\ldots,B_m,A_1B_1$, $\ldots,
A_mB_m,A_1^2B_1$, $\ldots,A_m^{n-1}B_m]=\mathbb{R}^n.$ Since $\forall
P\in \mathbb{R}^{n\times r}$,
$\textit{Im}(A_iP)\subseteq\textit{Im}(A_i)$, we have that
$\textit{Im}[B_1,\ldots,B_m$, $A_1B_1,\ldots,A_mB_m,A_1^2B_1,\ldots, A_m^{n-1}B_m]$ $\subseteq\textit{Im}[A_1,A_2,\ldots,A_m,
B_1,B_2,\ldots,B_m]\subseteq\mathbb{R}^n$. Thus condition
$g$-rank $\mathcal{C}(A_1,\ldots,A_m,$ $B_1,\ldots,B_m)=n$ requires
that $\textit{Im}[A_1,A_2,\ldots,A_m,
B_1,B_2,\ldots, B_m]=\mathbb{R}^n$ and therefore $g$-rank
$[A_1,A_2,\ldots,A_m, B_1,B_2,\ldots,B_m]=n$. However, if there is
$S$-$dilation$ in the colored union graph $\mathcal{\tilde{G}}$, by
Lemma \ref{lem5}, $g$-rank  $[A_1,A_2,\ldots,A_m,
B_1,B_2,\ldots,B_m]<n$. Consequently, the switched linear system
(\ref{eq2}) is not structurally controllable.

\textit{Sufficiency}: 
The general idea in the sufficiency proof is that we will assume
that the two graphical conditions in the theorem hold. Then a
contradiction will be found such that it is impossible that switched
linear system (\ref{eq2}) is structurally uncontrollable.

Before proceeding to switched linear system (\ref{eq2}), firstly,
consider a structured linear system:
\begin{equation}\label{r1}
\dot{x}(t)=Ax(t)+Bu(t)
\end{equation}
It is well known that system (\ref{r1}) is structurally controllable
if and only if there exists a numerical realization
$(\tilde{A},\tilde{B})$, such that rank
$(sI-\tilde{A},\tilde{B})=n,\forall s\in \mathbb{C}$. Otherwise, the
PBH test (\cite{TK}) states that system (\ref{r1}) is uncontrollable
if and only if for every numerical realization, there exists a row
vector $q\neq0$ such that $q\tilde{A}=s_0q,s_0\in \mathbb{C}$ and
$q\tilde{B}=0$, where rank $(s_0I-\tilde{A},\tilde{B})<n$.

On one hand, if for every numerical realization rank
$(sI-\tilde{A},\tilde{B})=n,\forall s\in \mathbb{C}\setminus \{0\}$,
then the uncontrollability of system (\ref{r1}) implies necessarily
that for every numerical realization there exists a vector $q\neq0$
such that $q\tilde{A}=0$ and $q\tilde{B}=0$.

On the other hand, Lemma 14.1 of \cite{K} states that, if in the
digraph associated to (\ref{r1}), every state vertex is an end
vertex of a stem (accessible), then $g$-rank $(sI-A,B)=n,\forall
s\in \mathbb{C}\setminus \{0\}$, which means that for almost all
numerical realization $(\tilde{A},\tilde{B})$, rank
$(sI-\tilde{A},\tilde{B})=n,\forall s\in \mathbb{C}\setminus \{0\}$.

Now considering switched linear system (\ref{eq2}), assume that the
two conditions in Theorem \ref{the2} are satisfied. Due to Lemma
14.1 of \cite{K}, as all the parameters of matrices
$A_1,\ldots,A_m,B_1$, $\ldots,B_m$ are assumed to be free, the
condition \textit{(i)} of Theorem \ref{the2} implies that, for
almost all vector values $\bar{u}=(\bar{u}_1,\ldots,\bar{u}_m)$, we
have $g$-rank
$(sI-(\bar{u}_1A_1+\ldots+\bar{u}_mA_m),(\bar{u}_1B_1+\ldots+\bar{u}_mB_m))=n,
\forall s\neq 0$. On the other hand, if switched linear system
(\ref{eq2}) is structurally uncontrollable, then for all constant
values, $\bar{u}=(\bar{u}_1,\ldots,\bar{u}_m)$, linear systems
defined by matrices $(\bar{A},\bar{B})$ are also uncontrollable,
where $\bar{A}=\sum^m_{i=1}\bar{u}_iA_i$ and
$\bar{B}=\sum^m_{i=1}\bar{u}_iB_i$. We write the numerical
realization of $(\bar{A},\bar{B})$ as
$(\tilde{\bar{A}},\tilde{\bar{B}})$. This is due to the fact that
for all constant values $\bar{u}$,
$\textit{Im}(\mathcal{C}(\bar{A},\bar{B})\subseteq
\textit{Im}(\mathcal{C}(A_1,\ldots,A_m,B_1,\ldots,B_m)).$ Therefore,
if the switched linear system is structurally uncontrollable, since
for almost all $\bar{u}=(\bar{u}_1,\ldots,\bar{u}_m)$, $g$-rank
$(sI-(\bar{u}_1A_1+\ldots+\bar{u}_mA_m),(\bar{u}_1B_1+\ldots+\bar{u}_mB_m))=n,
\forall s\neq 0$, we have that for every numerical realization
matrix pair $(\tilde{\bar{A}},\tilde{\bar{B}})$, there exists a
nonzero vector $q$ such that $q\tilde{\bar{A}}=0$ and
$q\tilde{\bar{B}}=0$. Since this statement is true for almost all
the values $\bar{u}=(\bar{u}_1,\ldots,\bar{u}_m)$, we have that for
almost all $n\cdot m$-tuple values
$\bar{u}^j=(\bar{u}^j_1,\ldots,\bar{u}^j_m), j=1,\ldots,n\cdot m$,
we can find nonzero vectors $q_j$ such that the following holds:
\begin{equation}\label{r2}
\left\{\begin{array}{clc}
   \sum^m_{i=1}\bar{u}^j_iq_j\tilde{A}_i=0, \hfill j=1,\ldots,n\cdot
   m\\
    \sum^m_{i=1}\bar{u}^j_iq_j\tilde{B}_i=0. \hfill  j=1,\ldots,n\cdot m\\
\end{array}\right.\end{equation}
Obviously, there can not exist more than $n$ linear independent
vectors $q_j$. Let us denote $q_1,q_2,\ldots,q_n$ the vectors such
that $span$ $(q_1,q_2,\ldots,q_{n\cdot m})\subseteq$\\ $span$
$(q_1,q_2,\ldots,q_n)$ (we can renumber the vectors if necessary).
All the vectors $q_j,j=n+1,\ldots,n\cdot m$ are linear combinations
of $q_1,q_2,\ldots,q_n$. Therefore, system (\ref{r2}) contains the
following equations:
\begin{equation}\label{r3}
\left\{\begin{array}{clc}
   \sum^n_{k=1}\sum^m_{i=1}a^j_{i,k}(\bar{u})q_k\tilde{A}_i=0 \hfill ~~~~j=1,\ldots,n\cdot m  \\
   \sum^n_{k=1}\sum^m_{i=1}a^j_{i,k}(\bar{u})q_k\tilde{B}_i=0 \hfill ~~~~j=1,\ldots,n\cdot m  \\
\end{array}\right.\end{equation}
where $a^j_{i,k}(\bar{u})$ are linear functions of
$\bar{u}^j,j=1,\ldots,n\cdot m$. Since system (\ref{r2}) is
satisfied for almost all the values, we can find
$\bar{u}^j,j=1,\ldots,n\cdot m$ such that
\begin{equation}\label{r4}
det\left[ {\begin{array}{*{20}c}
    a^1_{1,1}(\bar{u})& a^1_{1,2}(\bar{u})&\ldots&a^1_{m,n}(\bar{u})  \\
   a^2_{1,1}(\bar{u})& a^2_{1,2}(\bar{u})&\ldots&a^2_{m,n}(\bar{u})  \\
   \vdots&\vdots&\vdots&\vdots\\
   a^{n\cdot m}_{1,1}(\bar{u})& a^{n\cdot m}_{1,2}(\bar{u})&\ldots&a^{n\cdot m}_{m,n}(\bar{u})
   \end{array}} \right]\neq 0.\nonumber
\end{equation}
In this case, the only solution of (\ref{r3}) is
$q_k\tilde{A}_1=\ldots=q_k\tilde{A}_m=q_k\tilde{B}_1=\dots=q_k\tilde{B}_m=0,$
$k=1,\ldots,n$. Obviously, if the switched linear system is
structurally
uncontrollable, then 
vector $q_k,k=1,\ldots,n$ is nonzero. Consequently, switched linear
system (\ref{eq2}) is structurally uncontrollable only if for every
numerical realization there exists at least one nonzero vector $q$
such that $qA_1=\ldots=qA_m=qB_1=\dots=qB_m=0$. However, if
condition \textit{ii} of Theorem \ref{the2} is satisfied, then
$g$-rank $[A_1,\ldots,A_m,B_1,\ldots,B_m]=n$ and therefore, for at
least one numerical realization, there does not exist a vector
$q\neq 0$ such that $qA_1=\ldots=qA_m=qB_1=\dots=qB_m=0$. Hence, the
two conditions are sufficient to ensure the structural
controllability of switched linear system (\ref{eq2}).
\end{pf}

Actually, using the terminologies $`dilation'$ and $`S$-$dilation'$
as graphic criteria is not so numerically efficient. For example, to
check the second condition of Theorem \ref{the2}, we need to test
for all possible vertex subsets to see whether there exist
$S$-$dilation$ in the colored union graph or not. Consequently, we
will adopt another notion $`S$-$disjoint$ $edges$' to form a more
numerically efficient graphic interpretation of structural
controllability.
\begin{definition}\label{r5} In the colored union graph
$\mathcal{\tilde{G}}$, consider $k$ edges
$e_1=(v_1,v'_1),e_2=(v_2,v'_2),\ldots,e_k=(v_k,v'_k)$. We define for
$i=1,\ldots,k,$ $S_i$ as the set of integers $j$ such that
$v_j=v_i$, i.e., $S_i=\{1\leq j\leq k|v_j=v_i\}$. These $k$ edges
$e_1,e_2,\ldots,e_k$ are $S$-$disjoint$ if the following two
conditions are satisfied:
\begin{enumerate}
\item[i)]edges $e_1,e_2,\ldots,e_k$ have distinct end vertices,
\item[ii)] for $i=1,\ldots,k$, $S_i=\{i\}$ or there exist $r$
distinct integers $i_1,i_2,\ldots,i_r$ such that $e_{j_1}\in
\mathcal{I}_{i_1},e_{j_2}\in \mathcal{I}_{i_2},\ldots,e_{j_r}\in
\mathcal{I}_{i_r}$, where $j_1,j_2,\ldots,j_r$ are all the elements
of $S_i$.
\end{enumerate}
\end{definition}
Roughly speaking, $k$ edges are $S$-$disjoint$ if their end vertices
are all distinct and if all the edges which have the same begin
vertex can be associated to distinct indexes $i$. For this new
graphic property, the following lemma can be given:

\begin{lemma}\label{10} \rm Considering switched linear system (\ref{eq2}),
there exist $n$ $S$-$disjoint$ edges in associated colored union
graph $\mathcal{\tilde{G}}$ if and only if $[A_1,A_2,\ldots,A_m,
B_1,B_2,\ldots,B_m]$ has $g$-rank = $n$.
\end{lemma}
\begin{pf}
\textit{~~Necessity:} If there exist $n$ $S$-$disjoint$ edges in
$\mathcal{\tilde{G}}$, matrix $[A_1,A_2,\ldots,A_m,\\
B_1,B_2,\ldots,B_m]$ contains at least $n$ free parameters. Since
the $n$ $S$-$disjoint$ edges have distinct end vertices, the
corresponding $n$ free parameters lie on $n$ different rows.
Besides, the $n$ $S$-$disjoint$ edges have distinct begin vertices
or have same begin vertex that can be associated to distinct indexes
$i$. This implies that these $n$ free parameters lie on $n$
different columns. keep these $n$ free parameters and set all the
other free parameters to be zero. We can see that matrix
$[A_1,A_2,\ldots,A_m, B_1,B_2,\ldots,B_m]$ has following form:
 $\left[ {\begin{array}{*{20}c}
    0& \lambda_1&0&0&\ldots&0  \\
  0& 0&0& \lambda_2&\ldots&0\\
   \vdots&\vdots&\vdots&\vdots\\
   \lambda_n&0&0&0&\ldots&0
   \end{array}}. \right]$, which has $g$-rank = $n$.

\textit{Sufficiency}: From the Definition 12.3 and the following
discussions of \cite{K}, for a structured matrix $Q$, $g$-rank $Q$ =
$s$-rank $Q$. where $s$-rank of $Q$ is defined as the maximal number
of free parameters that no two of which lie on the same row or
column. If matrix $[A_1,A_2,\ldots,A_m, B_1,B_2,\ldots,B_m]$ has
$g$-rank = $n$, it follows that there exists $n$ free parameters
from $n$ different rows
, which implies that the corresponding $n$ edges have different end
vertices, from $n$ different columns, which implies that these $n$
edges start from different vertices or start from same vertices but
can be associated to different indexes. Hence condition that matrix
has $g$-rank = $n$ is sufficient to ensure existence of $n$
$S$-$disjoint$ edges.
\end{pf}
With the above definition and lemma, another necessary and
sufficient condition for structural controllability of system
(\ref{eq2}) can be proposed here:
 \begin{theorem}\label{the6}\rm
Switched linear system (\ref{eq2}) with graphic representations
$\mathcal{G}_i$, $i\in \{1,\ldots,m\}$, is structurally controllable
if and only if its colored union graph $\mathcal{\tilde{G}}$
satisfies the following two conditions:
\begin{enumerate}
\item[i)]there is no nonaccessible vertex in the colored union graph
$\mathcal{\tilde{G}}$,
\item[ii)]there exist $n$ $S$-$disjoint$ edges in the colored union graph $\mathcal{\tilde{G}}$.
\end{enumerate}
\end{theorem}
\begin{pf}
Lemma \ref{lem5} and Lemma \ref{10} show that there exist $n$
$S$-$disjoint$ edges in the colored union graph
$\mathcal{\tilde{G}}$ if and only if there is no $S$-$dilation$ in
$\mathcal{\tilde{G}}$. Then this theorem follows immediately.
\end{pf}
\subsection{Computation Complexity of The Proposed Criteria }
Compared with condition using $`S$-$dilation'$, this condition using
`$S$-$disjoint$ edges' does not require to check all the vertex
subsets, which is a more efficient criterion. The maximal number of
`$S$-$disjoint$ edges' can be calculated using bipartite graphs. For
example, we can use the algorithm in \cite{SV}, which allows to
compute the cardinality of maximum matching into a bipartite graph.
A bipartite graph is a graph whose vertices can be divided into two
disjoint sets $\mathcal{U}$ and $\mathcal{W}$ such that every edge
connects a vertex in $\mathcal{U}$ to one in $\mathcal{W}$. To build
a bipartite graph in directed subgraph $\mathcal{G}_i(\mathcal{V}_i,
\mathcal{I}_i)$, what we need to do is adding some vertices and
making $\mathcal{U}_i=\{v\in \mathcal{V}_i|\exists (v,v')\in
\mathcal{I}_i\}$, which implies that cardinality $|\mathcal{U}_i|$
equals to the number of nonzero columns in matrix $[A_i, B_i]$.
Besides, $\mathcal{W}_i=\mathcal{X}_i$, i.e., the state vertex set.
Then it follows that the maximum matching in this bipartite graph is
the same as the maximal $S$-$disjoint$ edge set in
$\mathcal{G}_i(\mathcal{V}_i, \mathcal{I}_i)$. According to
definition of $S$-$disjoint$ edges, the beginning vertex from
different subgraphs should be differentiated when building the
bipartite graph for colored union graph $\mathcal{\tilde{G}}$.
Therefore for the bipartite graph of $\mathcal{\tilde{G}}$,
$\mathcal{U}=\{v|\exists (v,v')\in \mathcal{I}_i, i=1,2,\ldots,m\}$,
which implies that cardinality $|\mathcal{U}|$ equals to the number
of nonzero columns in matrix $[A_1,A_2,\ldots,A_m,
B_1,B_2,\ldots,B_m]$. And $\mathcal{W}=\mathcal{X}$, i.e., the state
vertex set. Similarly, the maximum matching in this bipartite graph
is the same as the maximal $S$-$disjoint$ edge set in colored union
graph. Therefore the complexity order of algorithm using method in
\cite{SV} is $O(\sqrt{p+n}\cdot q)$, where $q$ is the number of
edges in colored union graph, i.e., the number of free parameters in
all system matrices, $p$ is the number of nonzero columns in matrix
$[A_1,A_2,\ldots,A_m, B_1,B_2,\ldots,B_m]$ and $n$ is number of
state variables. Compared with condition (ii) of Theorem \ref{the6},
condition (i) of Theorem \ref{the6} is easier to check. We have to
look for paths which connect each state vertex with one of the input
vertex. This is a standard task of algorithmic graph theory. For
example, depth-first search or breadth-first search algorithm for
traversing a graph can be adopted and the complexity order is
$O(|V|+|E|)$, where $|V|$ and $|E|$ are cardinalities of vertex set
and edge set in union graph.

\subsection{Illustrative Examples}

Consider a switched linear system with two subsystems as depicted by
the graphic topologies in Fig. 1(a)-(b). In colored union graph
$\mathcal{\tilde{G}}$ (Fig. 1(d)), thin lines represent edges from
subgraph (a) and thick lines represent the edges from subgraph (b).
It turns out that the colored union graph $\mathcal{\tilde{G}}$ has
no nonaccessible vertex and no $S$-$dilation$. Besides, the three
edges are $S$-$disjoint$ edges since they have different end
vertices and one edge begins at vertex 3 and two edges begin at
vertex 0 but they come from different subsystems. \vspace{0.6in}
\setlength{\unitlength}{0.0215in}

\begin{picture}(80,35)

\put(45,20){\circle{2.5}}\put(75,20){\circle{2.5}}\put(75,20){\circle{2.5}}
\put(95,20){\circle{2.5}}\put(125,20){\circle{2.5}}\put(75,20){\circle{2.5}}
\put(145,20){\circle{2.5}}\put(175,20){\circle{2.5}}\put(75,20){\circle{2.5}}
\put(45,35){\circle{2.5}} \put(95,35){\circle{2.5}}
\put(145,35){\circle{2.5}}

\put(45,50){\circle{2.5}}
\put(95,50){\circle{2.5}}
\put(145,50){\circle{2.5}}

\put(173,20){\vector(-1,0){27}}
\put(173,21){\vector(-1,1){28}}
\put(145,21){\vector(0,1){13.5}} \put(73,20){\vector(-1,0){27}}
\put(95,21){\vector(0,1){13.5}}
\put(123,21){\vector(-1,1){28}}

\put(42,20){\makebox(0,0)[c]{$3$}}\put(92,20){\makebox(0,0)[c]{$3$}}\put(142,20){\makebox(0,0)[c]{$3$}}
\put(79,20){\makebox(0,0)[c]{$0$}}\put(129,20){\makebox(0,0)[c]{$0$}}\put(179,20){\makebox(0,0)[c]{$0$}}
\put(42,50){\makebox(0,0)[c]{$1$}}\put(92,50){\makebox(0,0)[c]{$1$}}\put(142,50){\makebox(0,0)[c]{$1$}}
\put(42,35){\makebox(0,0)[c]{$2$}}\put(92,35){\makebox(0,0)[c]{$2$}}\put(142,35){\makebox(0,0)[c]{$2$}}

\put(60,10){\makebox(-2,0){$(a)$}}
\put(110,10){\makebox(-2,0){$(b)$}}
\put(160,10){\makebox(-2,0){$(c)$}}
\put(210,10){\makebox(-2,0){$(d)$}}
\put(195,20){\circle{2.5}}\put(225,20){\circle{2.5}}\put(195,35){\circle{2.5}}\put(195,50){\circle{2.5}}
\put(192,20){\makebox(0,0)[c]{$3$}}\put(229,20){\makebox(0,0)[c]{$0$}}\put(192,35){\makebox(0,0)[c]{$2$}}

\put(224,20){\vector(-1,0){27}} \thicklines
\put(224,20){\vector(-1,1){29}} \put(195,20){\vector(0,1){13.5}}
\put(192,50){\makebox(0,0)[c]{$1$}}
\put(130,-5){\makebox(15,0)[c]{{\footnotesize Fig. 1. Switched
linear system with two subsystems }}}

\end{picture}
\vspace{0.1in}

According to Theorem \ref{the2} or \ref{the6}, the switched linear
system is structurally controllable. On the other hand, the system
matrices of each subsystem of corresponding subgraph are:
\begin{equation}\label{124}
A_1= \left[ {\begin{array}{*{20}c}
   0&0 &0  \\
   0&0 & 0  \\
   0&0 &0
\end{array}} \right],~~B_1=\left[ {\begin{array}{*{20}c}
   0  \\
   0 \\
   \lambda_1
\end{array}} \right];
A_2= \left[ {\begin{array}{*{20}c}
   0&0 & 0 \\
   0&0 & \lambda_2  \\
   0&0 & 0
\end{array}} \right],~~B_2=\left[ {\begin{array}{*{20}c}
  \lambda_3 \\
   0  \\
   0
\end{array}} \right].\nonumber  \end{equation}
controllability matrix (\ref{eq3}) can be calculated and can be
shown to have $g$-$rank$=3. In addition, there exist a $dilation$ in
union graph Fig. 1(c), which shows that the condition in Theorem
\ref{the1} is not necessary for structural controllability.

In the following example, we will consider a real control object
with switched linear system model: A PWM-Driven Boost Converter
\cite{WL} as illustrated in Fig. 2.

In this electrical network, $L$ is the inductance, $C$ the
capacitance, $R$ the load resistance, and $e_S(t)$ the source
voltage. With this converter, the source voltage $e_S(t)$ can be
transformed into a higher voltage $e_C(t)$ over the load $R$. The
switch $s(t)$, which is supposed to have two states, namely, 0 and
1, is controlled by a PWM device.
By introducing the normalized variables $\tau=t/T$
, $L_1=L/T$, and $C_1=C/T$, the dynamics for the Boost converter are
described as follows:
\begin{equation}
\begin{array}{clc}
   \dot{e}_C(\tau)=\frac{-1}{RC_1}{e}_C(\tau)+(1-s(\tau))\frac{1}{C_1}i_L(\tau),\\
  \dot{i}_L(\tau)=-(1-s(\tau))\frac{1}{L_1}{e}_C(\tau)+s(\tau))\frac{1}{L_1}e_S(\tau),\\
\end{array}\end{equation}
Let $x_1 = e_C$, $x_2 = i_L$, $u = e_S$ $\sigma= s + 1,$ then the
system dynamics can be described as:
\begin{equation}\dot{x}=A_{\sigma}x+B_{\sigma}u,\sigma\in\{1,2\}
\end{equation}
where:
\begin{equation}
A_1= \left[ {\begin{array}{*{20}c}
   -\frac{1}{RC_1}&\frac{1}{C_1}  \\
   -\frac{1}{L_1}&0 \\

\end{array}} \right],~~B_1=\left[ {\begin{array}{*{20}c}
   0\\
   0\\

\end{array}} \right];
A_2= \left[ {\begin{array}{*{20}c}
  -\frac{1}{RC_1}&0  \\
   0&0 \\
\end{array}} \right],~~B_2=\left[ {\begin{array}{*{20}c}
      0\\
    \frac{1}{L_1}\\
\end{array}} \right].\nonumber  \end{equation}
Modeling this system using independent parameter and zero elements,
we have that
\begin{equation}
A_1= \left[ {\begin{array}{*{20}c}
   \lambda_1& \lambda_2\\
    \lambda_3&0 \\

\end{array}} \right],~~B_1=\left[ {\begin{array}{*{20}c}
   0\\
   0\\

\end{array}} \right];
A_2= \left[ {\begin{array}{*{20}c}
   \lambda_4&0  \\
   0&0 \\
\end{array}} \right],~~B_2=\left[ {\begin{array}{*{20}c}
      0\\
     \lambda_5\\
\end{array}} \right].\nonumber  \end{equation}
\begin{figure}
\includegraphics[width=0.67\textwidth]{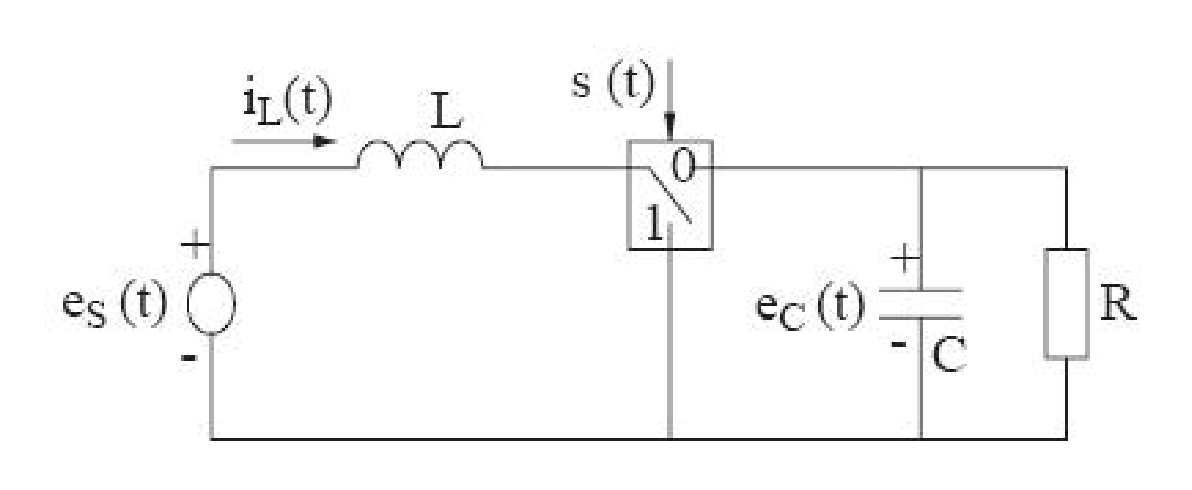}\centering
\\
\centerline{\footnotesize Fig. 2. The boost Converter} \label{fig2}
\end{figure}

%
\begin{figure}
\includegraphics[width=0.7\textwidth]{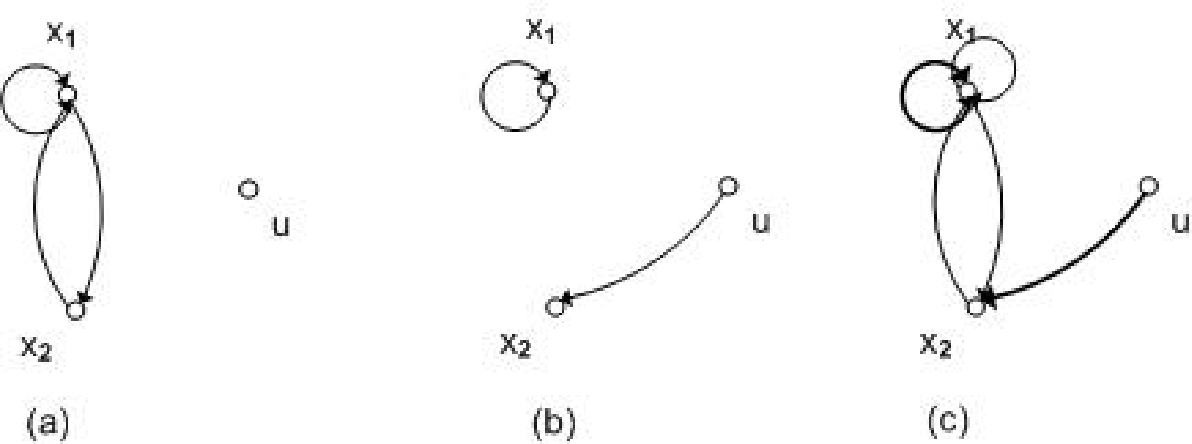}\centering
\\
\centerline{\footnotesize Fig. 3. Switched linear system with two
subsystems} \label{fig4}
\end{figure}

The two subsystems are depicted by the graphic topologies in Fig.
3(a)-(b). In colored union graph $\mathcal{\tilde{G}}$ (Fig. 3(c)),
thin lines represent edges from subgraph (a) and thick lines
represent the edges from subgraph (b). It turns out that the colored
union graph $\mathcal{\tilde{G}}$ has no nonaccessible vertex and no
$S$-$dilation$. Besides, the edge starting from $x_2$ and ending at
$x_1$ with index $(a)$ together with the edge starting from $u$ and
ending at $x_2$ with index $(b)$ consist of two $S$-$disjoint$ edges
since they have different starting and ending vertices. According to
the results obtained above, this switched electrical network is
structurally controllable and similarly the rank condition can be
checked that it has full $g$-$rank$ 2.

Form the above example, we can see that in some real applications
there are some dependent parameters among subsystems (since under
our independent case, the structural controllability holds for
almost all values of the free parameters, the dependent case can be
treated as a further extension but will not belittle the
significance of results obtained above). For further investigation
purpose, next we will use examples to illustrate that the dependence
among system parameters will make some edges `useless' or
`excessive' in judging the structural controllability. See the
following switched linear system first $A_1= \left[
{\begin{array}{*{20}c}
   0&0  \\
   0&0 \\
\end{array}} \right],~~B_1=\left[ {\begin{array}{*{20}c}
   \lambda_1\\
    \lambda_2\\
\end{array}} \right];
A_2= \left[ {\begin{array}{*{20}c}
  0&0  \\
   0&0 \\
\end{array}} \right],~~B_2=\left[ {\begin{array}{*{20}c}
      \lambda_3\\
    \lambda_4\\
\end{array}} \right].\nonumber$
According to Theorem \ref{the2} or \ref{the6}, this system is
structurally controllable. However, if dependent parameters are
considered, see the following switched linear system (a linear
system actually) $A_1= \left[ {\begin{array}{*{20}c}
   0&0  \\
   0&0 \\
\end{array}} \right],~~B_1=\left[ {\begin{array}{*{20}c}
   \lambda_1\\
    \lambda_2\\
\end{array}} \right];
A_2= \left[ {\begin{array}{*{20}c}
  0&0  \\
   0&0 \\
\end{array}} \right],~~B_2=\left[ {\begin{array}{*{20}c}
      \lambda_1\\
    \lambda_2\\
\end{array}} \right].\nonumber$
The dependence of all the parameters in matrix $B_1$ and $B_2$ makes
this system not structurally controllable and the results in Theorem
\ref{the2} or \ref{the6} not hold, even though it would be
structurally controllable if the parameters in $B_2$ are replaced
with $\lambda_3$ and $\lambda_4$ or simply remove $\lambda_1$ or
$\lambda_2$ in the second subsystem.

\section{Conclusions and Future Work}
In this paper, structural controllability for switched linear
systems has been investigated. Combining the knowledge in the
literature of switched linear systems and graph theory, several
graphic necessary and sufficient conditions for the structurally
controllability of switched linear systems have been proposed. 
These graphic interpretations provide us a better understanding on
how the graphic topologies of switched linear systems will influence
or determine the structural controllability of switched linear
systems. This shows us a new perspective that we can design the
switching algorithm to make the switched linear system structurally
controllable conveniently just having to make sure some properties
of the corresponding graph (union or colored union graph) are kept
during the switching process. In this paper, the parameters in
different subsystem models are assumed to be independent. A more
general assumption is that some free parameters remain the same
among different subsystems switching, i.e., dependence among
subsystems. It turns out that our necessary and sufficient condition
derived here would be a necessary condition under this dependence
assumption. Besides, our result can be treated as basic starting
point for exploring the structural controllability of switched
nonlinear systems: using Lie algebra or transfer function methods to
get full characterization for controllability of switched non-
linear system, then try to interpret each condition into graphic one
and ¡¥nally combine these conditions together to get graphic
interpretations for structural controllability for switched
nonlinear system. To obtain a full characterization for the
dependent case or switched nonlinear case needs further
investigation.

\renewcommand{\baselinestretch}{1}

\end{document}